\documentclass[11pt]{article}

\usepackage[paperwidth=199.8mm,paperheight=297mm,centering,hmargin=20mm,vmargin=25mm]{geometry}

\usepackage{cite}

\usepackage[utf8]{inputenc}
\usepackage{authblk}
\usepackage{amsmath, amssymb, amsthm, bbm, mathtools, nccmath, mathrsfs}
\usepackage{url, hyperref}
\hypersetup{
    colorlinks=true,
    citecolor=blue,
    linkcolor=blue,
    filecolor=blue,
    urlcolor=blue,
    breaklinks=true}
\usepackage{float, caption, subcaption}
\theoremstyle{plain} 
    \newtheorem{theorem}{Theorem}
    \newtheorem{lemma}[theorem]{Lemma}
    \newtheorem*{lemma*}{Lemma}
    \newtheorem{fact}[theorem]{Fact}
    
    \newtheorem{proposition}[theorem]{Proposition}

\theoremstyle{definition}
    \newtheorem{definition}[theorem]{Definition}
\theoremstyle{remark}
    \newtheorem{remark}[theorem]{Remark}
\renewcommand{\mathbf}{\boldsymbol}

\renewcommand{\varepsilon}{\epsilon}

\renewcommand{\geq}{\geqslant}

\renewcommand{\leq}{\leqslant}
\renewcommand{\tilde}{\widetilde}
\newcommand*{\set}[1]{\mathcal{#1}} 
\newcommand*{\sets}[1]{\mathfrak{#1}} 
\newcommand*{\rv}[1]{\mathsf{#1}} 
\newcommand*{\sys}[1]{\mathsf{#1}} 
\newcommand*{\syss}[1]{\mathbf{\mathsf{#1}}} 
\newcommand*{\channel}[1]{\mathcal{#1}} 

\newcommand*{\tensor}{\otimes}

\DeclareMathOperator{\tr}{tr}

\newcommand{\hilbert}{\mathcal{H}}

\newcommand{\reals}{\mathbb{R}}

\newcommand{\naturals}{\mathbb{N}}
\newcommand{\LinOp}{\mathscr{L}}
\newcommand{\HermOp}{\mathscr{H}}
\newcommand{\PSD}{\HermOp_{\scalebox{0.7}{+}}}
\newcommand{\PD}{\HermOp_{\scalebox{0.7}{++}}}
\newcommand{\density}{\mathscr{D}}
\newcommand{\herm}{\dagger}
\newcommand{\defeq}{\coloneqq}
\newcommand{\defas}{\eqqcolon}

\newcommand{\mle}{\preccurlyeq}
\DeclarePairedDelimiterX{\ket}[1]{\lvert}{\rangle}{#1}
\DeclarePairedDelimiterX{\bra}[1]{\langle}{\rvert}{#1}
\DeclarePairedDelimiterX{\spr}[2]{\langle}{\rangle}{#1\delimsize\vert#2}
\DeclarePairedDelimiterX{\proj}[1]{\lvert}{\rvert}{#1\delimsize\rangle\!\delimsize\langle#1}
\DeclarePairedDelimiterX{\floor}[1]{\lfloor}{\rfloor}{#1}
\DeclarePairedDelimiterX{\ceil}[1]{\lceil}{\rceil}{#1}
\DeclarePairedDelimiterX{\abs}[1]{\lvert}{\rvert}{#1}
\DeclarePairedDelimiterX{\norm}[1]{\lVert}{\rVert}{#1}
\DeclarePairedDelimiterX{\size}[1]{\lvert}{\rvert}{#1}
\DeclarePairedDelimiterX{\infdiv}[2]{(}{)}{#1\delimsize\Vert#2}
\DeclarePairedDelimiterX{\infdivc}[3]{(}{)}{#1\delimsize\Vert#2\delimsize\vert#3}
\DeclarePairedDelimiterX{\inner}[2]{\langle}{\rangle}{#1,#2}
\newcommand{\Div}{\mathbb{D}} 
\newcommand{\uDiv}{D}
\newcommand{\pDiv}[1]{D_{\Petz, #1}}
\newcommand{\rDiv}[1]{{D}_{\Sand, #1}}
\newcommand{\hDiv}[1]{D_{\Hypo, #1}}
\newcommand{\mDiv}{D_{\scriptscriptstyle  \rm M}}
\newcommand{\mrDiv}[1]{D_{{\scriptscriptstyle  \rm M}, #1}}
\newcommand{\supDiv}{\Div^{\uparrow}}

\newcommand{\infDiv}{\Div^{\downarrow}}
\newcommand{\iinfDiv}{\Div^{\downarrow\downarrow}}
\newcommand{\iinfmDiv}{\mDiv^{\downarrow\downarrow}}
\newcommand{\iinfmrDiv}[1]{\mrDiv{#1}^{\downarrow\downarrow}}
\newcommand{\Meas}{{\scriptscriptstyle \rm M}}

\newcommand{\Renyi}{R\'{e}nyi }

\newcommand{\Sand}{{\scriptscriptstyle  \rm S}}
\newcommand{\Petz}{{\scriptscriptstyle  \rm P}}
\newcommand{\Hypo}{{\scriptscriptstyle  \rm H}}

\newcommand{\new}{\text{\rm new}}

\newcommand{\CPTP}{\text{\rm CPTP}}

\newcommand{\PER}{\text{\rm PERM}}

\newcommand{\CP}{\text{\rm CP}}

\newcommand{\spec}{\text{\rm spec}}
\newcommand{\poly}{\text{\rm poly}}

\newcommand{\minimax}{{\operatorname{mx}}}

\ExplSyntaxOn
\NewDocumentCommand{\multiadjustlimits}{m}
 {
  \group_begin:
  \multiadjustlimits_measure:n { #1 }
  \multiadjustlimits_print:n { #1 }
  \group_end:
 }

\tl_new:N  \l__multiadjustlimits_operator_tl
\tl_new:N  \l__multiadjustlimits_limit_tl

\cs_new_protected:Nn \multiadjustlimits_measure:n
 {
  \clist_map_function:nN { #1 } \__multiadjustlimits_measure:n
 }
\cs_new_protected:Nn \__multiadjustlimits_measure:n
 {
  \__multiadjustlimits_measure:NNn #1
 }
\cs_new_protected:Nn \__multiadjustlimits_measure:NNn
 {
  \tl_put_right:Nn \l__multiadjustlimits_operator_tl { #1 }
  \tl_put_right:Nn \l__multiadjustlimits_limit_tl { #3 }
 }

\cs_new_protected:Nn \multiadjustlimits_print:n
 {
  \clist_map_function:nN { #1 } \__multiadjustlimits_print:n
 }
\cs_new_protected:Nn \__multiadjustlimits_print:n
 {
  \__multiadjustlimits_print:NNn #1
 }
\cs_new_protected:Nn \__multiadjustlimits_print:NNn
 {
  \mathop { \vphantom{\l__multiadjustlimits_operator_tl} \mathopen{} #1 }
  \limits
  \sb{ \vphantom{\cramped{\l__multiadjustlimits_limit_tl}} #3 }
 }

\ExplSyntaxOff
\newcommand\ie{\textit{i.e.}}
\newcommand\eg{\textit{e.g.}}
\newcommand\cf{\textit{cf.}}

\newcommand\aka{a.k.a.~}

\usepackage{framed}
\usepackage[dvipsnames]{xcolor} \usepackage{color}
\definecolor{shadecolor}{rgb}{0.9,0.9,0.9}
\begin{document}
\title{\Large \bfseries Quantum channel discrimination against jammers}
\author[1]{Kun Fang}
\author[2]{Michael X. Cao}
\affil[1]{\small{School of Data Science, The Chinese University of Hong Kong, Shenzhen,
Guangdong, 518172, China}}
\affil[2]{\small{Institute for Quantum Information, RWTH Aachen University, Aachen, Germany}}
\date{\vspace{-1em}}
\maketitle
\begin{abstract}
We study the problem of quantum channel discrimination between two channels with an adversary input party (\aka a jammer).
This setup interpolates between the best-case channel discrimination as studied by (Wang~\&~Wilde,~2019) and the worst-case channel discrimination as studied by (Fang,~Fawzi,~\&~Fawzi,~2025), thereby generalizing both frameworks.
To address this problem, we introduce the notion of minimax channel divergence and establish several of its key mathematical properties. 
We prove the Stein's lemma in this new setting, showing that the optimal type-II error exponent in the asymptotic regime under parallel strategies is characterized by the regularized minimax channel divergence.
\end{abstract}
\section{Introduction}

Channel discrimination is a fundamental task in both classical and quantum information theory~\cite{hayashi2009discrimination}. 
The goal is to distinguish between two channels—conditional distributions in the classical case, or completely positive trace-preserving (CPTP) maps in the quantum case—by preparing suitable inputs and observing the corresponding outputs. Quantum channel discrimination plays a pivotal role in a variety of quantum information tasks, providing insights into numerous protocols and applications~\cite{dariano_using_2001,chiribella_memory_2008,hayashi_discrimination_2009,harrow_adaptive_2010,wilde2020amortized,Cooney2016,wang2019resource,pirandola_fundamental_2019,WW2019,bergh_parallelization_2024}. These range from foundational studies—such as exploring the quantum advantage of entanglement~\cite{piani2009all,takagi_operational_2019,bae_more_2019,skrzypczyk_robustness_2019}—to practical applications in quantum communication (e.g., estimating channel capacities~\cite{Wang2012,datta_smooth_2013,Wang2019,wang2019converse,fang2021geometric,fang2025towards}), quantum sensing (e.g., quantum reading and quantum illumination~\cite{pirandola_advances_2018,zhuang_ultimate_2020}), and even quantum biology~\cite{spedalieri_detecting_2020,pereira2020quantum}.

Previous studies on quantum channel discrimination have focused exclusively on the conventional ``best-case'' scenario~\cite{acin2001statistical,duan2007entanglement,chiribella_memory_2008,piani2009all,duan2009perfect,bae_more_2019,takagi_operational_2019,skrzypczyk_robustness_2019,fang2020chain,zhuang_ultimate_2020,bavaresco2021strict,debry2023experimental,sugiura2024power}, in which the tester has full control over input states, intermediate operations, and final measurements. More recently, a ``worst-case'' adversarial framework was introduced~\cite{fang2025adversarial}, motivated by applications such as quantum device verification~\cite{zhu2019efficient}, where a state preparation device from an untrusted manufacturer is expected to produce a specific resource state (e.g., a Bell state, magic state, or coherent state), but may instead output arbitrary junk states if faulty or maliciously designed.  In the adversarial setting, the roles of the tester and adversary are fundamentally different, resulting in a competitive game: the tester designs the measurement, while the adversary controls the choice of input states and potentially intermediate updates.

In this paper, we turn to channel discrimination in the presence of quantum-enabled jammers, a new discrimination framework that interpolates the best-case~\cite{wang2019resource} and worst-case~\cite{fang2025adversarial} scenarios.
Our focus is on the asymmetric error exponent (also known as the Stein's exponent) in the asymptotic regime. 
The study of channels with jammers -- \ie, channels subject to adversarial inputs -- dates back to the late 1950s~\cite{blackwell1959capacity, blackwell1960capacities}. 
Early work considered classical adversaries, leading to models such as compound channels~\cite{wolfowitz1961coding} and arbitrarily varying channels (AVCs)~\cite{ahlswede1970note} (see also~\cite{csiszar2011information}). 
These models have since been generalized to the quantum setting, including arbitrarily varying classical-quantum and quantum channels~\cite{ahlswede2007classical, ahlswede2013quantum, berta2017entanglement, boche2017entanglement,Dasgupta2025universal}, where the jammer is effectively classical;
and, most notably, fully quantum arbitrarily varying channels (FQAVCs)~\cite{boche2018fully, belzig2024fully}. 
However, prior work has focused almost exclusively on capacity results (\ie, coding tasks), rather than the more fundamental problem of channel discrimination.

Beyond its theoretical interest, our study is also operationally motivated. 
As already pointed out in~\cite{fang2025adversarial}, the verification of distrusted quantum devices arises naturally in a variety of quantum information tasks, including quantum key distribution~\cite{Bennett1984, mayers1998quantum}, quantum adversarial learning~\cite{lu2020quantum}, quantum interactive proofs~\cite{vidick2016quantum}, and quantum state verification~\cite{zhu2019efficient}. 
Unlike the fully adversarial setting of~\cite{fang2025adversarial}, however, we consider devices that are \emph{partially} distrusted, in the sense that the tester retains some control via the choice of channel input. 
From a technical perspective, this task of discriminating  with adversarial inputs forms a competitive minimax game-theorectic framework and interpolates naturally between the conventional ``best-case'' setup~\cite{wang2019resource} (with a trivial jammer system) and the recent ``worst-case'' (adversarial) setup~\cite{fang2025adversarial} (with a trivial input system), thereby providing a unified generalization of both frameworks (see Figure~\ref{fig:parallel-parallel}).
\begin{figure}
    \centering
    \includegraphics[width=\linewidth]{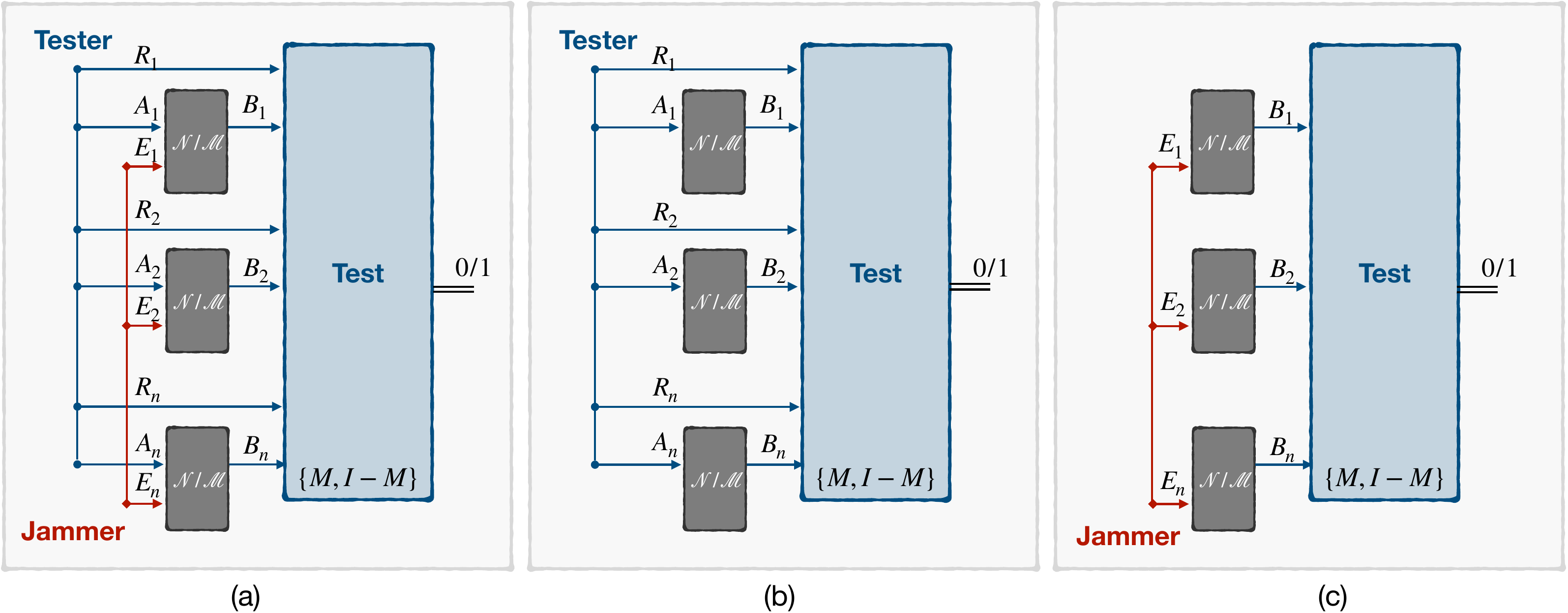}
    \caption{Comparison of different quantum channel discrimination setups via parallel strategies. (a): Tester plays against jammer (this work); (b) The best-case setup with a trivial jammer system~\cite{wang2019resource}; (c) The worst-case setup with a trivial input system~\cite{fang2025adversarial}.}
    \label{fig:parallel-parallel}
\end{figure}

Technically, the aforementioned task of quantum channel discrimination against jammers is formulated as a hypothesis testing problem:
A quantum channel with a jammer is a CPTP with dual input systems, \ie, $\channel{N}_{\sys{AE}\to\sys{B}} \in \CPTP(\sys{AE}:\sys{B})$, where $\sys{A}$ and $\sys{B}$ are the input and the output system, respectively, and $\sys{E}$ is the adversary/jammer system.
Under the one-shot setup, to distinguish between a pair of such channels, say, $\channel{N}_{\sys{AE}\to\sys{B}}$ and $\channel{M}_{\sys{AE}\to\sys{B}} \in \CPTP(\sys{AE}:\sys{B})$, the strategy can be described by a pair of input preparation $\rho_\sys{AR}\in\density(\hilbert_\sys{A}\tensor\hilbert_\sys{R})$ and measurement $0\mle M \mle I_\sys{B}$ on the output system, where $\sys{R}$ is some auxiliary system.
The type-I and type-II error can hence be defined as 
\begin{align}
\label{eq:def:type-1:error}
\alpha\infdivc*{\channel{N}_{\sys{AE}\to\sys{B}}}{\channel{M}_{\sys{AE}\to\sys{B}}}{\rho_\sys{AR}} &\defeq \sup_{\sigma_\sys{E}\in\density(\hilbert_\sys{E})} \tr\big( \channel{N}_{\sys{AE}\to\sys{B}}(\rho_\sys{AR}\tensor\sigma_\sys{E})\cdot (I_\sys{B} - M) \big), \\
\label{eq:def:type-2:error}
\beta\infdivc*{\channel{N}_{\sys{AE}\to\sys{B}}}{\channel{M}_{\sys{AE}\to\sys{B}}}{\rho_\sys{AR}} &\defeq \sup_{\sigma_\sys{E}\in\density(\hilbert_\sys{E})} \tr\big( \channel{M}_{\sys{AE}\to\sys{B}}(\rho_\sys{AR}\tensor\sigma_\sys{E})\cdot M \big),
\end{align}
respectively.
These definitions reflect a worst-case adversary, since the jammer is free to choose any input state on system $\sys{E}$.
For asymmetric hypothesis testing, we want to study the optimal type-II error probability
\begin{equation}\label{eq:def:opt:type-2:error}
\beta^{\sys{E}}_\epsilon\infdiv*{\channel{N}_{\sys{AE}\to\sys{B}}}{\channel{M}_{\sys{AE}\to\sys{B}}} \defeq
\inf_{\substack{0\mle M \mle I_\sys{B}\\ \rho_\sys{AR} \in \density(\hilbert_\sys{A}\tensor\hilbert_\sys{R})}}
\Big\{ \beta\infdivc*{\channel{N}_{\sys{AE}\to\sys{B}}}{\channel{M}_{\sys{AE}\to\sys{B}}}{\rho_\sys{AR}} 
    \Big\vert
    \alpha\infdivc*{\channel{N}_{\sys{AE}\to\sys{B}}}{\channel{M}_{\sys{AE}\to\sys{B}}}{\rho_\sys{AR}} \leq \epsilon \Big\}.
\end{equation}
In particular, we are interested in studying the error exponent,
\begin{equation}\label{eq:def:error:exponent}
\lim_{n\to\infty} -\frac{1}{n}\log{\beta^{\sys{E}}_\epsilon\infdiv*{\channel{N}_{\sys{AE}\to\sys{B}}^{\tensor n}}{\channel{M}_{\sys{AE}\to\sys{B}}^{\tensor n}}} = \ ?
\end{equation}
The above scenario is known as a parallel strategy; in contrast, sequential strategies could depend on intermediate outputs.

Towards solving the above problem, we develop the notion of minimax channel divergence together with an array of its interesting mathematical properties.
In particular, we prove the minimax property of the minimax channel divergence (Lemma~\ref{lem: jammer divergence minimax}),  which shows the equivalence of a few variant definitions. We also establish the optimal type-II error with our framework via the minimax channel divergence (Proposition~\ref{prop: beta DH relation jammer}).
We further show that the optimal jamming states can always be chosen permutation-invariant (Lemma~\ref{lem:mimmaxDiv:perm}), a fact that underpins the Stein's lemma we established later. 
We then turn to the measured (Rényi) minimax channel divergence, where we establish its super-additivity (Lemma~\ref{lem:superadditive:mrDiv}) and prove its asymptotic equivalence to the Umegaki minimax channel divergence (Lemma~\ref{lem: DM D jam infty}). 
These results form the technical foundation for our characterization of the Stein's exponent in the presence of quantum jammers (Theorem~\ref{thm:stein's lemma jammer}).

The rest of the paper is organized as follows.
In Section~\ref{sec: preliminaries}, we introduce necessary notations alongside with useful results on quantum divergences.
In Section~\ref{sec:divergence}, we introduce our notation of quantum minimax channel divergence, and prove several of its important properties.
In Section~\ref{sec:parallel:discrimination}, we apply the newly found tools to the problem of channel discrimination between channels with jammers and present the main theorem.
Section~\ref{sec:conclusion} concludes the paper.

\section{Preliminaries} \label{sec: preliminaries}
This section sets up the notation and recalls several divergences along with their main properties.
We begin with the following conventions and notation:
\begin{itemize}\setlength{\itemsep}{0pt}\setlength{\topsep}{0pt}
    \item Sets are denoted by calligraphic fonts, \eg, $\set{X}$ reads ``set $\set{X}$''.
    Random variables and quantum systems are denoted in sans serif fonts, \eg, $\rv{X}$ reads ``random variable $\rv{X}$'', and $\sys{A}$ reads ``system $\sys{A}$''.
    \item Vectors are denoted in boldface letters \eg, $\mathbf{x}$, and $\syss{A}$.
    In particular, we use the subscript and superscript to denote the starting and ending indexes of a vector.
    Namely, $\mathbf{x}_1^n\defeq (x_1, \ldots, x_n)$, and $\syss{A}_1^n\defeq (\sys{A}_1,\ldots, \sys{A}_n)$.
    \item Given a quantum system $\sys{A}$, its state space (which is a Hilbert space) is denoted by $\hilbert_\sys{A}$.
    We denote $\LinOp(\hilbert_\sys{A})$, $\HermOp(\hilbert_\sys{A})$, $\PSD(\hilbert_\sys{A})$, and $\PD(\hilbert_\sys{A})$ the set of all linear operators, Hermitian operator, positive semidefinite operators, and  positive definite operators on $\hilbert_\sys{A}$, respectively.
    Furthermore, we denote $\density(\hilbert_\sys{A})$ the set of density operators on $\hilbert_\sys{A}$.
    \item For two PSD operators on a same Hilbert space, \ie, $\rho,\sigma\in\PSD(\hilbert)$, $\rho\ll\sigma$ denotes that the support of $\rho$ is a subspace of the support of $\sigma$.
    \item We denote the set of all completely positive trace preserving (CPTP) maps from system $\sys{A}$ to system $\sys{B}$ by $\CPTP(\sys{A}:\sys{B})$, and the set of all completely positive (CP) maps (not necessary trace preserving) from system $\sys{A}$ to system $\sys{B}$ by $\CP(\sys{A}:\sys{B})$.
\end{itemize}

The remaining part of this section focuses on the concept of quantum divergence, \ie,
\begin{definition}[Quantum divergence]
    A functional $\Div: \density \times \PSD \to \reals$ is called a quantum divergence if it satisfies the data-processing inequality, \ie, 
	\begin{equation}
	\Div\infdiv*{\mathcal{E}(\rho)}{\mathcal{E}(\sigma)} \leq \Div\infdiv*{\rho}{\sigma}
	\end{equation}
	for any quantum channel, \ie, a CPTP map, $\mathcal{E}$, $\rho \in \density$, and $\sigma\in\PSD$. 
	In particular, a quantum divergence is said to satisfy the \emph{direct-sum property} if 
	\begin{equation}
	\Div\infdiv*{\sum_{x} p_\rv{X}(x)\cdot \proj{x} \tensor \rho^{(x)}}{\sum_{x} p_\rv{X} \proj{x} \tensor \sigma^{(x)}}
    = \sum_{x} p_\rv{X}(x)\cdot \Div\infdiv*{\rho^{(x)}}{\sigma^{(x)}}.
	\end{equation}
	for any random variable $\rv{X}$ and $\{\rho^{(x)}\}_{x},\{\sigma^{(x)}\}_{x}\subset\density$.
\end{definition}
\noindent 
It is noteworthy that the direct-sum property together with the data processing implies the joint convexity of such quantum divergence.
In the following, we will introduce several quantum divergences and their fundamental properties, which will be used throughout this work.

\begin{definition}[Useful quantum divergences]
	For any $\rho\in \density$, $\sigma \in \PSD$, and $\alpha \in (0,1) \cup (1,+\infty)$, there are following definitions:
        \vspace{-.5em}
	\begin{description}\setlength{\itemsep}{0pt}\setlength{\topsep}{0pt}
	    \item [Umegaki divergence~\cite{umegaki1954conditional}]
	          $\uDiv\infdiv*{\rho}{\sigma} \defeq \tr[\rho(\log \rho - \log \sigma)]$ if $\rho\ll\sigma$ and $+\infty$ otherwise.
	    \item [Petz \Renyi divergence~\cite{petz1986quasi}]
	          $\pDiv{\alpha}\infdiv*{\rho}{\sigma} \defeq \frac{1}{\alpha-1}\log\tr\left[\rho^\alpha\sigma^{1-\alpha}\right]$ if $\rho\ll\sigma$, and $+\infty$ otherwise.
	    \item [Sandwiched \Renyi divergence~\cite{muller2013quantum,wilde2014strong}]
	          $\rDiv{\alpha}\infdiv*{\rho}{\sigma} \defeq \frac{1}{\alpha-1}\log\tr\left[\sigma^{\frac{1-\alpha}{2\alpha}}\rho\sigma^{\frac{1-\alpha}{2\alpha}}\right]^\alpha$ if $\rho\ll\sigma$, and $+\infty$ otherwise.
	    \item [Hypothesis testing divergence]
	          For each $\epsilon \in [0,1]$, $\hDiv{\epsilon} \infdiv*{\rho}{\sigma} \defeq -\log \beta_{\epsilon}\infdiv*{\rho}{\sigma}$ where \\ $\beta_\epsilon\infdiv*{\rho}{\sigma} \defeq \min_{0\leq M \leq I} \left\{\tr[\sigma M]: \tr[\rho(I-M)] \leq \epsilon\right\}$.
	    \item [Measured divergence~\cite{donald1986relative,hiai1991proper}]
	          $\mDiv\infdiv*{\rho}{\sigma} \defeq \sup_{(\set{X},M)} D\infdiv*{P_{\rho,\Meas}}{P_{\sigma,\Meas}}$ where the optimization is over finite sets $\set{X}$ and positive operator valued measures $M$ on $\set{X}$ such that $M_x \geq 0$ and $\sum_{x \in \set{X}} M_x = I$, $P_{\rho,\Meas}$ is a measure on $\set{X}$ defined via the relation $P_{\rho,\Meas}(x) = \tr[M_x\rho]$ for any $x \in \set{X}$.  
	    \item [Measured \Renyi divergence~\cite{Berta2017}]
	          Similar to the above, $\mrDiv{\alpha}\infdiv*{\rho}{\sigma} \defeq \sup_{(\set{X},M)} D_{\alpha}\infdiv*{P_{\rho,\Meas}}{P_{\sigma,\Meas}}$ where $D_{\alpha}$ is the classical \Renyi divergence. 
	\end{description}
\end{definition}
In the following, we collect a list of useful facts from the literature.
\begin{fact}[{Combining \cite[Lemma 5]{cooney2016strong} and~\cite[Eq.~(2)]{wang2012one}}] \label{fact: DH petz sandwiched}
Let $\alpha \in (0,1)$ and $\epsilon\in(0,1)$. For any $\rho\in \density$ and $\sigma \in \PSD$, it holds that
\begin{align}
    \pDiv{\alpha}\infdiv*{\rho}{\sigma} + \frac{\alpha}{\alpha-1} \log \frac{1}{\epsilon}
    \leq \hDiv{\epsilon} \infdiv*{\rho}{\sigma}
    \leq \frac{1}{1-\epsilon}(\uDiv\infdiv*{\rho}{\sigma} + h(\epsilon)),
\end{align}
where $h(\epsilon):= -\epsilon \log \epsilon - (1-\epsilon)\log (1-\epsilon)$ is the binary entropy.
\end{fact}

\begin{fact}(\cite[Lemma 16]{fang2024generalized})\label{fact: DM and Sandwiched relation}
For any $\rho \in \density$ and $\sigma \in \PSD$, it holds that 
\begin{align}
    \mDiv\infdiv*{\rho}{\sigma} \leq D\infdiv*{\rho}{\sigma} \leq \mDiv\infdiv*{\rho}{\sigma} + 2\log \size{\spec(\sigma)},
\end{align}
where $\size{\spec(\sigma)}$ is the number of different eigenvalues of $\sigma$.
\end{fact}

The above lemma can be particularly useful when combined with the following fact. 
\begin{fact}[{\cite[Lemma A.1]{fawzi2021defining}}] \label{fact: permutation invariant spec}
Let $X$ be a permutation-invariant operator in $\LinOp(\hilbert^{\otimes n})$ with $\dim{\hilbert} = d$.
Then the number of mutually distinct eigenvalues of $X$ is upper bounded as $\size{\spec(X)} \leq (n+1)^{d} (n+d)^{d^2} = \poly(n)$.
\end{fact}


As the last part of the preliminary, we define the quantum divergences between two sets of quantum states, which is at the center of discussion of this paper.
\begin{definition}(Quantum divergence between two sets of states.)\label{def: divergence between two sets}
Given any quantum divergence $\Div\infdiv{\cdot}{\cdot}$, for any sets $\set{A}\subseteq \density$ and $\set{B}\subseteq \PSD$, the corresponding quantum divergence between these two sets of quantum states is defined as 
    \begin{align}
        \Div\infdiv*{\set{A}}{\set{B}} \defeq \inf_{\rho \in \set{A},\, \sigma \in \set{B}} \Div\infdiv*{\rho}{\sigma}. 
    \end{align}
\end{definition}
Note that if $\Div$ is lower semi-continuous (which is true for most quantum divergences of interest), and $\set{A}$ and $\set{B}$ are compact sets, the infimum in the above expression is always attainable and can thus be replaced by a minimization~\cite[Theorem~7.3.1]{kurdila2005convex}.
From a geometric perspective, this quantity characterizes the distance between two sets $\set{A}$ and $\set{B}$ under the ``distance metric'' $\Div$.
In particular, if $\set{A} = \{\rho\}$ is a singleton, we write $\Div\infdiv*{\rho}{\set{B}} \defeq \Div\infdiv*{\{\rho\}}{\set{B}}$.
For any two sequences of sets $\sets{A}=\{\set{A}_n\}_{n\in \naturals}$ and $\sets{B}=\{\set{B}_n\}_{n \in \naturals}$, the \emph{regularized} divergence is defined by 
\begin{align}
\Div^{\infty}\infdiv*{\sets{A}}{\sets{B}} \defeq \lim_{n \to \infty} \frac{1}{n} \Div\infdiv*{\set{A}_n}{\set{B}_n},
\end{align}
provided the limit on the right-hand side exists.

\section{Quantum minimax channel divergence} \label{sec:divergence}

In this section, we develop the notion of \emph{minimax channel divergence} for dual-input single-output quantum channels.
Such notion is very useful in dealing with quantum channels with an adversary input.
\begin{definition}
Let $\Div\infdiv{\cdot}{\cdot}$ be a quantum divergence, and let $\sys{A}$ and $\sys{B}$ be two quantum systems. 
For any $\channel{N} \in \CPTP(\sys{A}:\sys{B})$ and $\channel{M} \in \CP(\sys{A}:\sys{B})$.
The corresponding \emph{best-case channel divergence} is defined as~\cite{leditzky2018approaches}
\begin{align}
    \supDiv\infdiv*{\channel{N}_{\sys{A}\to\sys{B}}}{\channel{M}_{\sys{A}\to\sys{B}}}\defeq 
    \sup_{\rho_\sys{AR} \in \density(\hilbert_\sys{A}\tensor\hilbert_\sys{R})} \Div\infdiv*{\channel{N}_{\sys{A}\to\sys{B}}(\rho_\sys{AR})}{\channel{M}_{\sys{A}\to\sys{B}}(\rho_\sys{AR})},
\end{align}
where $\sys{R}$ is a reference system of arbitrary dimension.
The \emph{worst-case channel divergence} is defined as~\cite{fang2025adversarial}
\begin{align}
    \iinfDiv\infdiv*{\channel{N}_{\sys{A}\to\sys{B}}}{\channel{M}_{\sys{A}\to\sys{B}}} \defeq 
    \inf_{\sigma_\sys{A},\,\omega_\sys{A} \in \density(\hilbert_\sys{A})} \Div\infdiv*{\channel{N}_{\sys{A}\to\sys{B}}(\rho_\sys{A})}{\channel{M}_{\sys{A}\to\sys{B}}(\omega_\sys{A})}.
\end{align}
\end{definition}
\begin{definition}\label{def:minimaxDiv}
Let $\Div\infdiv{\cdot}{\cdot}$ be a quantum divergence, and let $\sys{A}$, $\sys{B}$ and $\sys{E}$ be quantum systems. 
For any $\channel{N} \in \CPTP(\sys{AE}:\sys{B})$ and $\channel{M} \in \CP(\sys{AE}:\sys{B})$, the corresponding \emph{minimax channel divergence} is defined as  
\begin{align}
    \Div^{\minimax}\infdiv*{\channel{N}_{\sys{AE}\to\sys{B}}}{\channel{M}_{\sys{AE}\to\sys{B}}} \defeq
    \adjustlimits \sup_{\rho_\sys{AR} \in \density(\hilbert_\sys{A}\tensor\hilbert_\sys{R})} \inf_{\sigma_\sys{E},\,\omega_\sys{E} \in \density(\hilbert_\sys{E})} \Div\infdiv*{\channel{N}_{\sys{AE}\to\sys{B}}(\rho_\sys{AR}\tensor\sigma_\sys{E})}{\channel{M}_{\sys{AE}\to\sys{B}}(\rho_\sys{AR}\tensor\omega_\sys{E})},
\end{align}
where $\sys{R}$ is a reference system of arbitrary dimension.
Furthermore, the induced regularized minimax channel divergence is defined as
\begin{align}
    \Div^{\minimax,\infty}\infdiv*{\channel{N}_{\sys{AE}\to\sys{B}}}{\channel{M}_{\sys{AE}\to\sys{B}}} \defeq
    \lim_{n\to \infty} \frac{1}{n} \Div^{\minimax}\infdiv*{\channel{N}_{\sys{AE}\to\sys{B}}^{\tensor n}}{\channel{M}_{\sys{AE}\to\sys{B}}^{\tensor n}},
\end{align}
provided that the limit exists.
\end{definition}
\begin{remark} \label{rem: jammer divergence sup of worst case}
The minimax channel divergence defined above interpolates between the best-case channel divergence (\ie, $\dim{\hilbert_\sys{E}}=1$) and the worst-case  channel divergence (\ie, $\dim{\hilbert_\sys{A}}=1$).
As an immediate observation, one can express the minimax channel divergence in terms of the worst-case channel divergence, \ie, 
\begin{align}\label{eq: DM jammer superadditivity tmp1}
    \Div^{\minimax}\infdiv*{\channel{N}_{\sys{AE}\to\sys{B}}}{\channel{M}_{\sys{AE}\to\sys{B}}} =
    \sup_{\rho_\sys{AR} \in \density(\hilbert_\sys{A} \otimes \hilbert_\sys{R})} \iinfDiv\infdiv*{\channel{N}^{\rho_\sys{AR}}_{\sys{E}\to\sys{B}}}{\channel{M}^{\rho_\sys{AR}}_{\sys{E}\to\sys{B}}}
\end{align}
where, for any $\channel{C}_{\sys{AE}\to\sys{B}} \in \CP(\sys{AE}:\sys{B})$, denote the induced channel $\channel{C}^{\rho_\sys{A}}_{\sys{A}\to\sys{B}}(\sigma_\sys{E})\defeq \channel{C}_{\sys{AE}\to\sys{B}}(\rho_\sys{A}\tensor\sigma_\sys{E})$.
\end{remark}

Note that we chose the infimum and the supremum in a specific order in the definition of $\Div^{\minimax}$, which, however, can be interchanged, as pointed out in the lemma below.
This is one of the reasons we coined the term ``minimax channel divergence'' as the minimax equality holds for this definition under suitable conditions.
\begin{shaded}
\begin{lemma}(Minimax property.)\label{lem: jammer divergence minimax}
Let $\Div\infdiv{\cdot}{\cdot}$ be a lower-semi-continuous quantum divergence with direct-sum property.
It holds for any $\channel{N} \in \CPTP(\sys{AE}:\sys{B})$ and $\channel{M} \in \CP(\sys{AE}:\sys{B})$ that 
\begin{align}
    \Div^{\minimax} 
    & \infdiv*{\channel{N}_{\sys{AE}\to\sys{B}}}{\channel{M}_{\sys{AE}\to\sys{B}}}\notag\\
    & \defeq
    \adjustlimits \sup_{\rho_\sys{AR} \in \density(\hilbert_\sys{A}\tensor\hilbert_\sys{R})} \inf_{\sigma_\sys{E},\,\omega_\sys{E} \in \density(\hilbert_\sys{E})} \Div\infdiv*{\channel{N}_{\sys{AE}\to\sys{B}}(\rho_\sys{AR}\tensor\sigma_\sys{E})}{\channel{M}_{\sys{AE}\to\sys{B}}(\rho_\sys{AR}\tensor\omega_\sys{E})} \label{eq: minimax 1}\\
    &= \adjustlimits \sup_{\rho_\sys{A} \in \density(\hilbert_\sys{A})} \inf_{\sigma_\sys{E},\,\omega_\sys{E} \in \density(\hilbert_\sys{E})} \Div\infdiv*{\channel{N}_{\sys{AE}\to\sys{B}}(\proj{\rho}_\sys{AR}\tensor\sigma_\sys{E})}{\channel{M}_{\sys{AE}\to\sys{B}}(\proj{\rho}_\sys{AR}\tensor\omega_\sys{E})} \label{eq: minimax 2}\\
    & = \adjustlimits \inf_{\sigma_\sys{E},\,\omega_\sys{E} \in \density(\hilbert_\sys{E})} \sup_{\rho_\sys{A} \in \density(\hilbert_\sys{A})} \Div\infdiv*{\channel{N}_{\sys{AE}\to\sys{B}}(\proj{\rho}_\sys{AR}\tensor\sigma_\sys{E})}{\channel{M}_{\sys{AE}\to\sys{B}}(\proj{\rho}_\sys{AR}\tensor\omega_\sys{E})} \label{eq: minimax 3}\\
    & = \adjustlimits \inf_{\sigma_\sys{E},\,\omega_\sys{E} \in \density(\hilbert_\sys{E})} \sup_{\rho_\sys{AR} \in \density(\hilbert_\sys{A}\tensor\hilbert_\sys{R})}  \Div\infdiv*{\channel{N}_{\sys{AE}\to\sys{B}}(\rho_\sys{AR}\tensor\sigma_\sys{E})}{\channel{M}_{\sys{AE}\to\sys{B}}(\rho_\sys{AR}\tensor\omega_\sys{E})} \label{eq: minimax 4}
\end{align}
where $\ket{\rho}_{\sys{AR}}$ is a purification of $\rho_\sys{A}$ with $\hilbert_\sys{R} = \hilbert_\sys{A}$.
\end{lemma}
\end{shaded}
\begin{proof}
    For any $\channel{C}_{\sys{AE}\to\sys{B}} \in \CP(\sys{AE}:\sys{B})$, let $\channel{C}^{\varrho_\sys{E}}_{\sys{A}\to\sys{B}}$ denotes the induced channel $\channel{C}^{\varrho_\sys{E}}_{\sys{A}\to\sys{B}}(\rho_\sys{A})\defeq \channel{C}_{\sys{AE}\to\sys{B}}(\rho_\sys{A}\tensor\varrho_\sys{E})$.
    For any fixed $\sigma_\sys{E},\omega_\sys{E} \in \density(\hilbert_\sys{E})$, we can write  
    \begin{align}
        \sup_{\rho_\sys{AR} \in \density(\hilbert_\sys{A}\tensor\hilbert_\sys{R})} & \Div\infdiv*{\channel{N}_{\sys{AE}\to\sys{B}}(\rho_\sys{AR}\tensor\sigma_\sys{E})}{\channel{M}_{\sys{AE}\to\sys{B}}(\rho_\sys{AR}\tensor\omega_\sys{E})} \notag \\
       & =  \sup_{\rho_\sys{AR} \in \density(\hilbert_\sys{A}\tensor\hilbert_\sys{R})} \Div\infdiv*{\channel{N}^{\sigma_\sys{E}}_{\sys{A}\to\sys{B}}(\rho_\sys{AR})}{\channel{M}^{\omega_\sys{E}}_{\sys{A}\to\sys{B}}(\rho_\sys{AR})}.
    \end{align}
    As a consequence of purification, data processing and the Schmidt decomposition, we can restrict $\rho_\sys{AR}$ to pure states, and restrict the  system $\sys{R}$ to be isomorphic to the system $\sys{A}$, \ie, 
    \begin{align}
         \sup_{\rho_\sys{AR} \in \density(\hilbert_\sys{A}\tensor\hilbert_\sys{R})} \Div& \infdiv*{\channel{N}^{\sigma_\sys{E}}_{\sys{A}\to\sys{B}}(\rho_\sys{AR})}{\channel{M}^{\omega_\sys{E}}_{\sys{A}\to\sys{B}}(\rho_\sys{AR})}\notag\\
         & =  \sup_{\ket{\phi} \in \hilbert_\sys{AR}} \Div\infdiv*{\channel{N}^{\sigma_\sys{E}}_{\sys{A}\to\sys{B}}(\proj{\phi}_\sys{AR})}{\channel{M}^{\omega_\sys{E}}_{\sys{A}\to\sys{B}}(\proj{\phi}_\sys{AR})}\\
         & = \sup_{\rho_\sys{A} \in \density(\hilbert_\sys{A})} \Div\infdiv*{\channel{N}^{\sigma_\sys{E}}_{\sys{A}\to\sys{B}}(\proj{\rho}_\sys{AR})}{\channel{M}^{\omega_\sys{E}}_{\sys{A}\to\sys{B}}(\proj{\rho}_\sys{AR})},
    \end{align}
    where the second line follows from the isometric invariance of quantum divergence. This implies that \eqref{eq: minimax 3} equals~\eqref{eq: minimax 4}.
    By~\cite[Lemma II.3]{leditzky2018approaches}, the direct-sum property implies that the objective function
   $\Div\infdiv*{\channel{N}_{\sys{AE}\to\sys{B}}(\proj{\rho}_\sys{AR}\tensor\sigma_\sys{E})}{\channel{M}_{\sys{AE}\to\sys{B}}(\proj{\rho}_\sys{AR}\tensor\omega_\sys{E})}$
    is concave in the marginal state $\rho_\sys{A}$. Furthermore, the combination of the direct-sum property and data processing ensures convexity in $\sigma_\sys{E}$ and $\omega_\sys{E}$. Therefore, by the minimax theorem~\cite{Sion1958}, we can interchange the order of the supremum and infimum, so that \eqref{eq: minimax 2} equals \eqref{eq: minimax 3}. It is also clear that \eqref{eq: minimax 4} is no smaller than \eqref{eq: minimax 1}, which in turn is no smaller than \eqref{eq: minimax 2}. This completes the proof.
\end{proof}
\begin{remark}\label{rem: jammer divergence sup of best case}
Lemma~\ref{lem: jammer divergence minimax} allows us to express the minimax channel divergence in terms of the best-case channel divergence as
\begin{align}
    \Div^{\minimax}\infdiv*{\channel{N}_{\sys{AE}\to\sys{B}}}{\channel{M}_{\sys{AE}\to\sys{B}}} =
    \inf_{\sigma_\sys{E},\, \omega_\sys{E} \in \density(\hilbert_\sys{E})} \supDiv\infdiv*{\channel{N}^{\sigma_\sys{E}}_{\sys{A}\to\sys{B}}}{\channel{M}^{\omega_\sys{E}}_{\sys{A}\to\sys{B}}}.
\end{align}
\end{remark}

In the following, we show an interesting property of the minimax channel divergence when applied to $n$-fold tensor product channels.
In particular, we prove that the infimum in Eq.~\eqref{eq: minimax 3} can be restricted to permutation-invariant states without changing the value of $n$-fold $\Div^\minimax$.
We refer to Appendix~\ref{app:lem:bipartite:covarian} for a more generalized version of this lemma.
\begin{shaded}
\begin{lemma}(Symmetry reduction.) \label{lem:mimmaxDiv:perm}
Let $\Div\infdiv{\cdot}{\cdot}$ be a lower-semi-continuous quantum divergence with direct-sum property. Let $\sys{A}$, $\sys{B}$ and $\sys{E}$ be quantum systems. 
For any $\channel{N} \in \CPTP(\sys{AE}:\sys{B})$ and $\channel{M} \in \CP(\sys{AE}:\sys{B})$, and positive integer $n\geq 2$, 
\begin{align}
\Div^{\minimax}& \infdiv*{\channel{N}_{\sys{AE}\to\sys{B}}^{\tensor n}}{\channel{M}_{\sys{AE}\to\sys{B}}^{\tensor n}} \notag \\
& = \adjustlimits \inf_{\substack{\sigma_{\syss{E}_1^n},\,\omega_{\syss{E}_1^n} \\ \in \PER(\hilbert_\sys{E}^{\tensor n})}} \sup_{\substack{\rho_{\syss{A}_1^n} \in \\ \density(\hilbert_\sys{A}^{\tensor n})}} \Div\infdiv*{\channel{N}_{\sys{AE}\to\sys{B}}^{\tensor n}(\proj{\rho}_{\syss{A}_1^n\syss{R}_1^n}\tensor\sigma_{\syss{E}_1^n})}{\channel{M}_{\sys{AE}\to\sys{B}}^{\tensor n}(\proj{\rho}_{\syss{A}_1^n\syss{R}_1^n}\tensor\omega_{\syss{E}_1^n})}\label{eq: minimax symmetry reduction 1}\\
& = \adjustlimits \inf_{\substack{\sigma_{\syss{E}_1^n},\,\omega_{\syss{E}_1^n} \\ \in \PER(\hilbert_\sys{E}^{\tensor n})}} \sup_{\substack{\rho_{\syss{A}_1^n} \in \\ \PER(\hilbert_\sys{A}^{\tensor n})}} \Div\infdiv*{\channel{N}_{\sys{AE}\to\sys{B}}^{\tensor n}(\proj{\rho}_{\syss{A}_1^n\syss{R}_1^n}\tensor\sigma_{\syss{E}_1^n})}{\channel{M}_{\sys{AE}\to\sys{B}}^{\tensor n}(\proj{\rho}_{\syss{A}_1^n\syss{R}_1^n}\tensor\omega_{\syss{E}_1^n})},\label{eq: minimax symmetry reduction 2}
\end{align}
where $\ket{\rho}_{\sys{AR}}$ is a purification of $\rho_\sys{A}$ with $\hilbert_\sys{R} = \hilbert_\sys{A}$.
\end{lemma}
\end{shaded}
\begin{proof}
    The key is to prove the following expression, as a function of $\sigma_{\syss{E}_1^n}$ and $\omega_{\syss{E}_1^n}$,
    \begin{equation}
    f_n(\sigma_{\syss{E}_1^n}, \omega_{\syss{E}_1^n}) \defeq \sup_{\rho_{\syss{A}_1^n} \in \density(\hilbert_\sys{A}^{\tensor n})} \underbrace{\Div\infdiv*{\channel{N}_{\sys{AE}\to\sys{B}}^{\tensor n}(\proj{\rho}_{\syss{A}_1^n\syss{R}_1^n}\tensor\sigma_{\syss{E}_1^n})}{\channel{M}_{\sys{AE}\to\sys{B}}^{\tensor n}(\proj{\rho}_{\syss{A}_1^n\syss{R}_1^n}\tensor\omega_{\syss{E}_1^n})}}_{\defas f_n(\sigma_{\syss{E}_1^n}, \omega_{\syss{E}_1^n}|\rho_{\syss{A}_1^n})}
    \end{equation}
    to be permutation invariant.
    For each permutation $\pi$ in the permutation group $\set{S}_n$, we define the unitary transformation $\pi_{\syss{A}_1^n}\in\LinOp(\hilbert_\sys{A}^{\tensor n})$ as 
    \begin{equation}
    \pi_{\syss{A}_1^n}: \ket{\psi_1}_{\sys{A}_1}\ket{\psi_2}_{\sys{A}_2}\cdots\ket{\psi_n}_{\sys{A}_n}
    \mapsto \ket{\psi_{\pi(1)}}_{\sys{A}_1}\ket{\psi_{\pi(2)}}_{\sys{A}_2}\cdots\ket{\psi_{\pi(n)}}_{\sys{A}_n}.
    \end{equation}
    We also abuse the notation a bit, and use $\pi_{\syss{A}_1^n}$ to denote the CPTP $\pi_{\syss{A}_1^n}: \rho_{\syss{A}_1^n} \mapsto \pi_{\syss{A}_1^n} \rho_{\syss{A}_1^n} \pi_{\syss{A}_1^n}^\herm$ as well.
    By carrying over the permutation of the output systems of the tensor of quantum channel to its input system (\cf~\cite[Eq.~(2)]{boche2018fully}~and~\cite[Theorem~3.3]{belzig2024fully}), we have 
    \begin{align}
        \pi_{\syss{B}_1^n}\left(\channel{N}_{\sys{AE}\to\sys{B}}^{\tensor n}(\proj{\rho}_{\syss{A}_1^n\syss{R}_1^n}\tensor\sigma_{\syss{E}_1^n})\right)
        &= \channel{N}_{\sys{AE}\to\sys{B}}^{\tensor n}\left(\pi_{\syss{A}_1^n\syss{E}_1^n}(\proj{\rho}_{\syss{A}_1^n\syss{R}_1^n}\tensor\sigma_{\syss{E}_1^n})\right) \\
        &= \channel{N}_{\sys{AE}\to\sys{B}}^{\tensor n}\left(\pi_{\syss{A}_1^n}(\proj{\rho}_{\syss{A}_1^n\syss{R}_1^n})\tensor\pi_{\syss{E}_1^n}(\sigma_{\syss{E}_1^n})\right)\\
        &= \channel{N}_{\sys{AE}\to\sys{B}}^{\tensor n}\left(\proj{\pi_{\syss{A}_1^n}\rho}_{\syss{A}_1^n\syss{R}_1^n}\tensor\pi_{\syss{E}_1^n}(\sigma_{\syss{E}_1^n})\right).
    \end{align}
    Similarly, the above equalities also hold for the channel $\channel{M}$.
    Since, quantum divergences are invariant under unitary transformations, we have 
    \begin{align}
        f_n&(\sigma_{\syss{E}_1^n}, \omega_{\syss{E}_1^n}|\rho_{\syss{A}_1^n}) \notag\\
        &= \Div\infdiv*{\pi_{\syss{B}_1^n}\left(\channel{N}_{\sys{AE}\to\sys{B}}^{\tensor n}(\proj{\rho}_{\syss{A}_1^n\syss{R}_1^n}\tensor\sigma_{\syss{E}_1^n})\right)}{\pi_{\syss{B}_1^n}\left(\channel{M}_{\sys{AE}\to\sys{B}}^{\tensor n}(\proj{\rho}_{\syss{A}_1^n\syss{R}_1^n}\tensor\omega_{\syss{E}_1^n})\right)} \\
        &=\Div\infdiv*{\channel{N}_{\sys{AE}\to\sys{B}}^{\tensor n}\left(\proj{\pi_{\syss{A}_1^n}\rho}_{\syss{A}_1^n\syss{R}_1^n}\!\tensor\!\pi_{\syss{E}_1^n}(\sigma_{\syss{E}_1^n})\right)}{\channel{M}_{\sys{AE}\to\sys{B}}^{\tensor n}\left(\proj{\pi_{\syss{A}_1^n}\rho}_{\syss{A}_1^n\syss{R}_1^n}\!\tensor\!\pi_{\syss{E}_1^n}(\omega_{\syss{E}_1^n})\right)}\\
        &= f_n(\pi_{\syss{E}_1^n}(\sigma_{\syss{E}_1^n}), \pi_{\syss{E}_1^n}(\omega_{\syss{E}_1^n})|\pi_{\syss{E}_1^n}(\rho_{\syss{A}_1^n})) .
    \end{align}
    Note that the set of all density operators on systems $\syss{A}_1^n$ is permutation invariant.
    Hence,
    \begin{align}
        f_n(\pi_{\syss{E}_1^n}(\sigma_{\syss{E}_1^n}), \pi_{\syss{E}_1^n}(\omega_{\syss{E}_1^n}))
        &= \sup_{\rho_{\syss{A}_1^n} \in \density(\hilbert_\sys{A}^{\tensor n})} f_n(\pi_{\syss{E}_1^n}(\sigma_{\syss{E}_1^n}), \pi_{\syss{E}_1^n}(\omega_{\syss{E}_1^n})|\rho_{\syss{A}_1^n})\\
        &= \sup_{\rho_{\syss{A}_1^n} \in \density(\hilbert_\sys{A}^{\tensor n})} f_n(\pi_{\syss{E}_1^n}(\sigma_{\syss{E}_1^n}), \pi_{\syss{E}_1^n}(\omega_{\syss{E}_1^n})|\pi_{\syss{E}_1^n}(\rho_{\syss{A}_1^n})) \\
        &= \sup_{\rho_{\syss{A}_1^n} \in \density(\hilbert_\sys{A}^{\tensor n})} f_n(\sigma_{\syss{E}_1^n}, \omega_{\syss{E}_1^n}|\rho_{\syss{A}_1^n}),
    \end{align}
    \ie, $f_n$ is permutation invariant.

    Therefore, for any optimizing pair $(\sigma_{\syss{E}_1^n}^\star, \omega_{\syss{E}_1^n}^\star)$, we can replace them by the average over all of its permutations, \ie,
    \begin{align}
    (\sigma_{\syss{E}_1^n}^\new, \omega_{\syss{E}_1^n}^\new)
    & \defeq \left(\sum_{\pi\in\set{S}_n} \size{\set{S}_n}^{-1}\cdot\pi_{\syss{E}_1^n}(\sigma_{\syss{E}_1^n}^\star), \sum_{\pi\in\set{S}_n} \size{\set{S}_n}^{-1}\cdot\pi_{\syss{E}_1^n}(\omega_{\syss{E}_1^n}^\star)\right)
    \end{align}
    which is a pair of permutation invariant states; hence, finishing the proof of Eq.~\eqref{eq: minimax symmetry reduction 1}. Since the induced channels $\channel{N}_{\sys{AE}\to\sys{B}}^{\tensor n}(\cdot\tensor\sigma_{\syss{E}_1^n})$ and $\channel{M}_{\sys{AE}\to\sys{B}}^{\tensor n}(\cdot\tensor\omega_{\syss{E}_1^n})$ are permutation invariant for $\sigma_{\syss{E}_1^n},\,\omega_{\syss{E}_1^n} \in \PER(\hilbert_\sys{E}^{\tensor n})$, we have Eq.~\eqref{eq: minimax symmetry reduction 2} by applying~\cite[Proposition II.4]{leditzky2018approaches}.
\end{proof}
In the remainder of this section, we focuses on specific choices of $\Div$, namely measured \Renyi divergence $\mrDiv{\alpha}$, measured divergence $\mDiv$ and Umegaki divergence $\uDiv$, for further exploration of the properties of $\Div^\minimax$.
\begin{shaded}
\begin{lemma}(Super-additivity.)
\label{lem:superadditive:mrDiv}
Let $\alpha \in (0,+\infty)$.
For any $\channel{N}^{(i)}_{\sys{AE}\to\sys{B}} \in \CPTP(\sys{A}_i\sys{E}_i:\sys{B}_i)$ and $\channel{M}^{(i)}_{\sys{AE}\to\sys{B}} \in \CP(\sys{A}_i\sys{E}_i:\sys{B}_i)$ with $i\in \{1,2\}$,  
\begin{equation}\label{eq:superadditive:mrDiv}
\begin{aligned}
    \mrDiv{\alpha}^\minimax\infdiv*{\channel{N}^{(1)}_{\sys{A}_1\sys{E}_1\to\sys{B}_1}\tensor \channel{N}^{(2)}_{\sys{A}_2\sys{E}_2\to\sys{B}_2}}{\channel{M}^{(1)}_{\sys{A}_1\sys{E}_1\to\sys{B}_1}\tensor \channel{M}^{(2)}_{\sys{A}_2\sys{E}_2\to\sys{B}_2}}
    \geq \hspace{100pt}\\
    \mrDiv{\alpha}^\minimax\infdiv*{\channel{N}^{(1)}_{\sys{A}_1\sys{E}_1\to\sys{B}_1}}{\channel{M}^{(1)}_{\sys{A}_2\sys{E}_2\to\sys{B}_2}} +     \mrDiv{\alpha}^\minimax\infdiv*{\channel{N}^{(2)}_{\sys{A}_1\sys{E}_1\to\sys{B}_1}}{\channel{M}^{(2)}_{\sys{A}_2\sys{E}_2\to\sys{B}_2}}.
\end{aligned}\end{equation}
\end{lemma}
\end{shaded}
\begin{proof}
The proof is done via the following chain of inequalities.
\begin{align}
\text{LHS of~\eqref{eq:superadditive:mrDiv}}
    \label{eq:superadditive:mrDiv:1}
    & = \sup_{\rho_{\syss{A}_1^2}} \iinfmrDiv{\alpha}\infdiv*{\left(\channel{N}^{(1)}\otimes \channel{N}^{(2)}\right)^{\rho_{\syss{A}_1^2}}_{\syss{E}_1^2\to\syss{B}_1^2}}{\left(\channel{M}^{(1)}\otimes \channel{M}^{(2)}\right)^{\rho_{\syss{A}_1^2}}_{\syss{E}_1^2\to\syss{B}_1^2}}\\
    \label{eq:superadditive:mrDiv:2}
    & \geq \sup_{\rho_{\sys{A}_1}, \rho_{\sys{A}_2}} \iinfmrDiv{\alpha}\infdiv*{\left(\channel{N}^{(1)}\otimes \channel{N}^{(2)}\right)^{\rho_{\sys{A}_1}\tensor\rho_{\sys{A}_2}}_{\syss{E}_1^2\to\syss{B}_1^2}}{\left(\channel{M}^{(1)}\otimes \channel{M}^{(2)}\right)^{\rho_{\sys{A}_1}\tensor\rho_{\sys{A}_2}}_{\syss{E}_1^2\to\syss{B}_1^2}}\\
    \label{eq:superadditive:mrDiv:3}
    & = \sup_{\rho_{\sys{A}_1}, \rho_{\sys{A}_2}} \iinfmrDiv{\alpha}\infdiv*{\channel{N}^{(1), \rho_{\sys{A}_1}}_{\sys{E}_1\to\sys{B}_1}\otimes \channel{N}^{(2), \rho_{\sys{A}_2}}_{\sys{E}_2\to\sys{B}_2}}{\channel{M}^{(1), \rho_{\sys{A}_1}}_{\sys{E}_1\to\sys{B}_1}\otimes \channel{M}^{(2), \rho_{\sys{A}_2}}_{\sys{E}_2\to\sys{B}_2}}\\
    \label{eq:superadditive:mrDiv:4}
    & \geq \sup_{\rho_{\sys{A}_1}, \rho_{\sys{A}_2}} \iinfmrDiv{\alpha}\infdiv*{\channel{N}^{(1), \rho_{\sys{A}_1}}_{\sys{E}_1\to\sys{B}_1}}{\channel{M}^{(1), \rho_{\sys{A}_1}}_{\sys{E}_1\to\sys{B}_1}} + \iinfmrDiv{\alpha}\infdiv*{\channel{N}^{(2), \rho_{\sys{A}_2}}_{\sys{E}_2\to\sys{B}_2}}{\channel{M}^{(2), \rho_{\sys{A}_2}}_{\sys{E}_2\to\sys{B}_2}}\\
    \label{eq:superadditive:mrDiv:5}
    & = \text{RHS of~\eqref{eq:superadditive:mrDiv}},
\end{align}
where~\eqref{eq:superadditive:mrDiv:1} and~\eqref{eq:superadditive:mrDiv:5} follow from Remark~\ref{rem: jammer divergence sup of worst case} and Lemma~\ref{lem: jammer divergence minimax}; \eqref{eq:superadditive:mrDiv:2} is a simple restriction on the feasible set; \eqref{eq:superadditive:mrDiv:3} follows directly from the definition of the induced channels; and \eqref{eq:superadditive:mrDiv:4} follows by applying twice the chain rule of the measured channel divergence~\cite[Lemma~2]{fang2025adversarial}.
\end{proof}
\begin{remark}\label{rem: lim sup DM alpha jam}
    As an immediate consequence of Lemma~\ref{lem:superadditive:mrDiv}, we have 
    \begin{align}
    \mrDiv{\alpha}^{\minimax,\infty}\infdiv*{\channel{N}_{\sys{AE}\to\sys{B}}}{\channel{M}_{\sys{AE}\to\sys{B}}} \defeq \lim_{n\to \infty} \frac{1}{n} \mrDiv{\alpha}^{\minimax}\infdiv*{\channel{N}_{\sys{AE}\to\sys{B}}^{\tensor n}}{\channel{M}_{\sys{AE}\to\sys{B}}^{\tensor n}}
    = \sup_{n\in\naturals} \frac{1}{n} \mrDiv{\alpha}^{\minimax}\infdiv*{\channel{N}_{\sys{AE}\to\sys{B}}^{\tensor n}}{\channel{M}_{\sys{AE}\to\sys{B}}^{\tensor n}},
    \end{align}
    where the second equality follows from Fekete's lemma.
\end{remark}
\begin{shaded}
\begin{lemma}(One-shot continuity.)\label{lem: DM alpha jam continuity}
For any $\channel{N}\in\CPTP(\sys{AE}:\sys{B})$ and $\channel{M}\in\CP(\sys{AE}:\sys{B})$, 
\begin{align}
    \sup_{\alpha \in (0,1)} \mrDiv{\alpha}^{\minimax}\infdiv*{\channel{N}_{\sys{AE}\to\sys{B}}}{\channel{M}_{\sys{AE}\to\sys{B}}} = \mDiv^{\minimax}\infdiv*{\channel{N}_{\sys{AE}\to\sys{B}}}{\channel{M}_{\sys{AE}\to\sys{B}}}.
\end{align}
\end{lemma}
\end{shaded}
\begin{proof}
    The proof is done via the following chain of equalities.
    \begin{align}
        \label{eq:DM alpha jam continuity:1}
        \sup_{\alpha \in (0,1)} \mrDiv{\alpha}^{\minimax}\infdiv*{\channel{N}_{\sys{AE}\to\sys{B}}}{\channel{M}_{\sys{AE}\to\sys{B}}} & = \sup_{\alpha \in (0,1)} \sup_{\rho_\sys{A} \in \density(\hilbert_\sys{A})} \iinfmrDiv{\alpha}\infdiv*{\channel{N}^{\rho_\sys{A}}_{\sys{E}\to\sys{B}}}{\channel{M}^{\rho_\sys{A}}_{\sys{E}\to\sys{B}}}\\
        \label{eq:DM alpha jam continuity:2}
        & =  \sup_{\rho_\sys{A} \in \density(\hilbert_\sys{A})} \sup_{\alpha \in (0,1)} \iinfmrDiv{\alpha}\infdiv*{\channel{N}^{\rho_\sys{A}}_{\sys{E}\to\sys{B}}}{\channel{M}^{\rho_\sys{A}}_{\sys{E}\to\sys{B}}}\\
        \label{eq:DM alpha jam continuity:3}
        & = \sup_{\rho_\sys{A} \in \density(\hilbert_\sys{A})} \iinfmDiv\infdiv*{\channel{N}^{\rho_\sys{A}}_{\sys{E}\to\sys{B}}}{\channel{M}^{\rho_\sys{A}}_{\sys{E}\to\sys{B}}}\\
        \label{eq:DM alpha jam continuity:4}
        & = \mDiv^{\minimax}\infdiv*{\channel{N}_{\sys{AE}\to\sys{B}}}{\channel{M}_{\sys{AE}\to\sys{B}}},
    \end{align}
    where~\eqref{eq:DM alpha jam continuity:1} and~\eqref{eq:DM alpha jam continuity:4} follow from Remark~\ref{rem: jammer divergence sup of worst case}; \eqref{eq:DM alpha jam continuity:3} follows from the continuity of the worst-case channel divergence~\cite[Lemma~22]{fang2024generalized} applied to the image sets of the induced channels.
\end{proof}
\begin{shaded}
\begin{lemma}(Asymptotic continuity.)\label{lem: DM alpha jam infty continuity}
For any $\channel{N}\in\CPTP(\sys{AE}:\sys{B})$ and $\channel{M}\in\CP(\sys{AE}:\sys{B})$, 
\begin{align}
    \sup_{\alpha \in (0,1)} \mrDiv{\alpha}^{\minimax,\infty}\infdiv*{\channel{N}_{\sys{AE}\to\sys{B}}}{\channel{M}_{\sys{AE}\to\sys{B}}}
    = \mDiv^{\minimax,\infty}\infdiv*{\channel{N}_{\sys{AE}\to\sys{B}}}{\channel{M}_{\sys{AE}\to\sys{B}}}.
\end{align}
\end{lemma}
\end{shaded}
\begin{proof}
    The proof is done via Remark~\ref{rem: lim sup DM alpha jam} and Lemma~\ref{lem: DM alpha jam continuity} as follows.
    \begin{align}
        \sup_{\alpha \in (0,1)} \mrDiv{\alpha}^{\minimax,\infty}\infdiv*{\channel{N}_{\sys{AE}\to\sys{B}}}{\channel{M}_{\sys{AE}\to\sys{B}}} 
        & = \sup_{\alpha \in (0,1)} \sup_{n\in\naturals} \frac{1}{n} \mrDiv{\alpha}^{\minimax}\infdiv*{\channel{N}_{\sys{AE}\to\sys{B}}^{\tensor n}}{\channel{M}_{\sys{AE}\to\sys{B}}^{\tensor n}} \\
        & = \sup_{n\in\naturals}\ \frac{1}{n} \sup_{\alpha \in (0,1)} \mrDiv{\alpha}^{\minimax}\infdiv*{\channel{N}_{\sys{AE}\to\sys{B}}^{\tensor n}}{\channel{M}_{\sys{AE}\to\sys{B}}^{\tensor n}} \\
        & = \sup_{n\in\naturals}\ \frac{1}{n} \mDiv^{\minimax}\infdiv*{\channel{N}_{\sys{AE}\to\sys{B}}^{\tensor n}}{\channel{M}_{\sys{AE}\to\sys{B}}^{\tensor n}}\\
        & = \mDiv^{\minimax,\infty}\infdiv*{\channel{N}_{\sys{AE}\to\sys{B}}}{\channel{M}_{\sys{AE}\to\sys{B}}}. \qedhere
    \end{align}
\end{proof}
To finish this section, we establish the asymptotic equivalence of the measured minimax channel  divergence and the Umegaki minimax channel  divergence as follows.
\begin{shaded}
\begin{lemma}(Asymptotic equivalence.)\label{lem: DM D jam infty}
For any $\channel{N}\in\CPTP(\sys{AE}:\sys{B})$ and $\channel{M}\in\CP(\sys{AE}:\sys{B})$, 
    \begin{align}
        \mDiv^{\minimax,\infty}\infdiv*{\channel{N}_{\sys{AE}\to\sys{B}}}{\channel{M}_{\sys{AE}\to\sys{B}}} = \uDiv^{\minimax,\infty}\infdiv*{\channel{N}_{\sys{AE}\to\sys{B}}}{\channel{M}_{\sys{AE}\to\sys{B}}}.
    \end{align}
\end{lemma}
\end{shaded}
\begin{proof}
    Since that $\mDiv\infdiv{\cdot}{\cdot} \leq \uDiv\infdiv{\cdot}{\cdot}$ in general, it is clear that \begin{align}\mDiv^{\minimax,\infty}\infdiv*{\channel{N}_{\sys{AE}\to\sys{B}}}{\channel{M}_{\sys{AE}\to\sys{B}}} \leq \uDiv^{\minimax,\infty}\infdiv*{\channel{N}_{\sys{AE}\to\sys{B}}}{\channel{M}_{\sys{AE}\to\sys{B}}}.
    \end{align}
    It suffices to show the other direction.
    Starting with Remark~\ref{rem: lim sup DM alpha jam} and Lemma~\ref{lem: DM alpha jam continuity}, we have 
    \begin{align}
        &\hspace{15pt} \mDiv^{\minimax,\infty}\infdiv*{\channel{N}_{\sys{AE}\to\sys{B}}}{\channel{M}_{\sys{AE}\to\sys{B}}} \nonumber \\
        &= \sup_{n\in\naturals} \frac{1}{n} \adjustlimits \inf_{\substack{\sigma_{\syss{E}_1^n},\,\omega_{\syss{E}_1^n} \\ \in \density(\hilbert_\sys{E}^{\tensor n})}} \sup_{\substack{\rho_{\syss{A}_1^n} \in \\ \density(\hilbert_\sys{A}^{\tensor n})}} \mDiv^{\minimax}\infdiv*{\channel{N}_{\sys{AE}\to\sys{B}}^{\tensor n}(\proj{\rho}_{\syss{A}_1^n\syss{R}_1^n}\tensor\sigma_{\syss{E}_1^n})}{\channel{M}_{\sys{AE}\to\sys{B}}^{\tensor n}(\proj{\rho}_{\syss{A}_1^n\syss{R}_1^n}\tensor\omega_{\syss{E}_1^n})} \\
        \label{eq: DM D jam infty:3}
        &= \sup_{n\in\naturals} \frac{1}{n} \adjustlimits \inf_{\substack{\sigma_{\syss{E}_1^n},\,\omega_{\syss{E}_1^n} \\ \in \PER(\hilbert_\sys{E}^{\tensor n})}} \sup_{\substack{\rho_{\syss{A}_1^n} \in \\ \PER(\hilbert_\sys{A}^{\tensor n})}} \mDiv^{\minimax}\infdiv*{\channel{N}_{\sys{AE}\to\sys{B}}^{\tensor n}(\proj{\rho}_{\syss{A}_1^n\syss{R}_1^n}\tensor\sigma_{\syss{E}_1^n})}{\channel{M}_{\sys{AE}\to\sys{B}}^{\tensor n}(\proj{\rho}_{\syss{A}_1^n\syss{R}_1^n}\tensor\omega_{\syss{E}_1^n})} 
    \end{align}
    where, for~\eqref{eq: DM D jam infty:3}, we have used Lemma~\ref{lem:mimmaxDiv:perm}.
    In this case, the state $\channel{M}_{\sys{AE}\to\sys{B}}^{\tensor n}(\proj{\rho}_{\syss{A}_1^n\syss{R}_1^n}\tensor\omega_{\syss{E}_1^n})$ is permutation invariant.
    By Fact~\ref{fact: permutation invariant spec}, we have 
    \begin{align}
        \size{\spec\left(\channel{M}_{\sys{AE}\to\sys{B}}^{\tensor n}(\proj{\rho}_{\syss{A}_1^n\syss{R}_1^n}\tensor\omega_{\syss{E}_1^n})\right)} \leq (n+1)^d(n+d)^{d^2},
    \end{align}
    where $d = \dim(\hilbert_\sys{R}\tensor\hilbert_\sys{B}) = \dim(\hilbert_\sys{A}\tensor\hilbert_\sys{B}) = \dim(\hilbert_\sys{A})\cdot\dim(\hilbert_\sys{B})$.
    Apply Fact~\ref{fact: DM and Sandwiched relation}, we have 
    \begin{align}
        \eqref{eq: DM D jam infty:3} & \geq
        \sup_{n\in\naturals} \frac{1}{n} \adjustlimits \inf_{\substack{\sigma_{\syss{E}_1^n},\,\omega_{\syss{E}_1^n} \\ \in \PER(\hilbert_\sys{E}^{\tensor n})}} \sup_{\substack{\rho_{\syss{A}_1^n} \in \\ \PER(\hilbert_\sys{A}^{\tensor n})}} \\
        & \bigg\{\uDiv^{\minimax}\infdiv*{\channel{N}_{\sys{AE}\to\sys{B}}^{\tensor n}(\proj{\rho}_{\syss{A}_1^n\syss{R}_1^n}\tensor\sigma_{\syss{E}_1^n})}{\channel{M}_{\sys{AE}\to\sys{B}}^{\tensor n}(\proj{\rho}_{\syss{A}_1^n\syss{R}_1^n}\tensor\omega_{\syss{E}_1^n})} - 2 \log (n+1)^d(n+d)^{d^2} \bigg\},\notag
    \end{align}
    which concludes the proof, as the logarithm term vanishes as $n\to\infty$.
\end{proof}

\section{Quantum channel discrimination against jammers} \label{sec:parallel:discrimination}
Recall from the introduction, we are interested in the task of channel discrimination between two channels with jammers $\channel{N}_{\sys{AE}\to\sys{B}},\,\channel{M}_{\sys{AE}\to\sys{B}} \in \CPTP(\sys{AE}:\sys{B})$, where $\sys{A}$ and $\sys{B}$ are the input and the output system, respectively, and $\sys{E}$ is the adversary/jammer system.
The type-I and type-II error are described by 
\begin{align}
\alpha\infdivc*{\channel{N}_{\sys{AE}\to\sys{B}}}{\channel{M}_{\sys{AE}\to\sys{B}}}{\rho_\sys{AR}} &\defeq \sup_{\sigma_\sys{E}\in\density(\hilbert_\sys{E})} \tr\big( \channel{N}_{\sys{AE}\to\sys{B}}(\rho_\sys{AR}\tensor\sigma_\sys{E})\cdot (I_\sys{B} - M) \big), \tag{\ref{eq:def:type-1:error}}\\
\beta\infdivc*{\channel{N}_{\sys{AE}\to\sys{B}}}{\channel{M}_{\sys{AE}\to\sys{B}}}{\rho_\sys{AR}} &\defeq \sup_{\sigma_\sys{E}\in\density(\hilbert_\sys{E})} \tr\big( \channel{M}_{\sys{AE}\to\sys{B}}(\rho_\sys{AR}\tensor\sigma_\sys{E})\cdot M \big), \tag{\ref{eq:def:type-2:error}}
\end{align}
respectively, where input preparation $\rho_\sys{AR}\in\density(\hilbert_\sys{A}\tensor\hilbert_\sys{R})$ and measurement $0\mle M \mle I_\sys{B}$ on the output system fully describes the discriminating strategy, and where $\sys{R}$ is some auxiliary system.
The optimal type-II error probability is defined as 
\begin{equation}
\beta^{\sys{E}}_\epsilon\infdiv*{\channel{N}_{\sys{AE}\to\sys{B}}}{\channel{M}_{\sys{AE}\to\sys{B}}} \defeq
\inf_{\substack{0\mle M \mle I_\sys{B}\\ \rho_\sys{AR} \in \density(\hilbert_\sys{A}\tensor\hilbert_\sys{R})}}
\Big\{ \beta\infdivc*{\channel{N}_{\sys{AE}\to\sys{B}}}{\channel{M}_{\sys{AE}\to\sys{B}}}{\rho_\sys{AR}} 
    \Big\vert
    \alpha\infdivc*{\channel{N}_{\sys{AE}\to\sys{B}}}{\channel{M}_{\sys{AE}\to\sys{B}}}{\rho_\sys{AR}} \leq \epsilon \Big\},
    \tag{\ref{eq:def:opt:type-2:error}}
\end{equation}
and we are interested in the error exponent
\begin{equation}
\lim_{n\to\infty} -\frac{1}{n}\log{\beta^{\sys{E}}_\epsilon\infdiv*{\channel{N}_{\sys{AE}\to\sys{B}}^{\tensor n}}{\channel{M}_{\sys{AE}\to\sys{B}}^{\tensor n}}} =\ ? \tag{\ref{eq:def:error:exponent}}
\end{equation}

As stated at the beginning of Section~\ref{sec:divergence}, we link the operational quantity as in~\eqref{eq:def:opt:type-2:error} to the minimax channel divergence as in the following proposition.
\begin{shaded}
\begin{proposition}\label{prop: beta DH relation jammer}
For any $\channel{N}\in\CPTP(\sys{AE}:\sys{B})$ and $\channel{M}\in\CP(\sys{AE}:\sys{B})$,
\begin{align}
    -\log{\beta^{\sys{E}}_\epsilon\infdiv*{\channel{N}_{\sys{AE}\to\sys{B}}}{\channel{M}_{\sys{AE}\to\sys{B}}}}  = 
    \hDiv{\epsilon}^{\minimax}\infdiv*{\channel{N}_{\sys{AE}\to\sys{B}}}{\channel{M}_{\sys{AE}\to\sys{B}}}.
\end{align}
\end{proposition}
\end{shaded}
\begin{proof}
By definition~\ref{def:minimaxDiv}, we can write $\hDiv{\epsilon}^{\minimax}$ as 
\begin{align}
    \hDiv{\epsilon}^{\minimax}&\infdiv*{\channel{N}_{\sys{AE}\to\sys{B}}}
    {\channel{M}_{\sys{AE}\to\sys{B}}}\notag\\
    & \defeq \adjustlimits \sup_{\rho_\sys{AR} \in \density(\hilbert_\sys{AR}\tensor\hilbert_\sys{R})} \inf_{\sigma_\sys{E},\,\omega_\sys{E} \in \density(\hilbert_\sys{E})} \hDiv{\epsilon}\infdiv*{\channel{N}_{\sys{AE}\to\sys{B}}(\rho_\sys{AR}\tensor\sigma_\sys{E})}{\channel{M}_{\sys{AE}\to\sys{B}}(\rho_\sys{AR}\tensor\omega_\sys{E})}\\
    & = -\log{\adjustlimits \inf_{\rho_\sys{AR} \in \density(\hilbert_\sys{A}\tensor\hilbert_\sys{R})} \sup_{\sigma_\sys{E},\,\omega_\sys{E} \in \density(\hilbert_\sys{E})} \beta_{\epsilon}\infdiv*{\channel{N}_{\sys{AE}\to\sys{B}}(\rho_\sys{AR}\tensor\sigma_\sys{E})}{\channel{M}_{\sys{AE}\to\sys{B}}(\rho_\sys{AR}\tensor\omega_\sys{E})}}.\label{eq: beta DH relation jammer tmp1}
\end{align}

Denote 
\begin{align}
    \beta^{\sys{E}}_\epsilon\infdivc*{\channel{N}_{\sys{AE}\to\sys{B}}}{\channel{M}_{\sys{AE}\to\sys{B}}}{\rho_{\sys{AR}}} \defeq
    \inf_{0\mle M \mle I_\sys{B}}
    \Big\{ \beta\infdivc*{\channel{N}_{\sys{AE}\to\sys{B}}}{\channel{M}_{\sys{AE}\to\sys{B}}}{\rho_\sys{AR}} 
    \Big\vert
    \alpha\infdivc*{\channel{N}_{\sys{AE}\to\sys{B}}}{\channel{M}_{\sys{AE}\to\sys{B}}}{\rho_\sys{AR}} \leq \epsilon \Big\}.
\end{align}
Immediately we have 
\begin{align}
    \beta^{\sys{E}}_\epsilon\infdiv*{\channel{N}_{\sys{AE}\to\sys{B}}}{\channel{M}_{\sys{AE}\to\sys{B}}}
    = \inf_{\rho_\sys{AR} \in \density(\hilbert_\sys{A}\tensor\hilbert_\sys{R})} \beta^{\sys{E}}_\epsilon\infdivc*{\channel{N}_{\sys{AE}\to\sys{B}}}{\channel{M}_{\sys{AE}\to\sys{B}}}{\rho_{\sys{AR}}}.\label{eq:beta DH relation jammer:1}
\end{align}

Now, let $\set{A}^{\rho_{\sys{AR}}}$ and $\set{B}^{\rho_{\sys{AR}}}$ denote the image of the induced channels $\channel{N}_{\sys{E}\to\sys{B}}^{\rho_{\sys{AR}}}$ and $\channel{M}_{\sys{E}\to\sys{B}}^{\rho_{\sys{AR}}}$ (see Remark~\ref{rem: jammer divergence sup of worst case} for the notation of the induced channels), respectively , \ie,
\begin{align}
    \set{A}^{\rho_{\sys{AR}}} &\defeq \left\{\channel{N}_{\sys{AE}\to\sys{B}}(\rho_{\sys{AR}}\tensor\varrho_\sys{E}) \middle\vert \varrho_\sys{E}\in\density(\hilbert_\sys{E}) \right\}, \\
    \set{B}^{\rho_{\sys{AR}}} &\defeq \left\{\channel{M}_{\sys{AE}\to\sys{B}}(\rho_{\sys{AR}}\tensor\varrho_\sys{E}) \middle\vert \varrho_\sys{E}\in\density(\hilbert_\sys{E}) \right\}.
\end{align}
It is straightforward that these two sets are convex; and moreover, we have 
\begin{align}
    \alpha\infdivc*{\channel{N}_{\sys{AE}\to\sys{B}}}{\channel{M}_{\sys{AE}\to\sys{B}}}{\rho_\sys{AR}} &= \alpha(\set{A}^{\rho_{\sys{AR}}},M) &&\text{ where } \alpha(\set{D},M)\defeq \sup\nolimits_{\varrho\in\set{D}} \tr\left(\varrho\cdot(I-M)\right),\\
    \beta\infdivc*{\channel{N}_{\sys{AE}\to\sys{B}}}{\channel{M}_{\sys{AE}\to\sys{B}}}{\rho_\sys{AR}} &= \beta(\set{B}^{\rho_{\sys{AR}}},M) &&\text{ where } \beta(\set{D},M)\defeq \sup\nolimits_{\varrho\in\set{D}} \tr\left(\varrho\cdot M\right).
\end{align}
Therefore,  
\begin{align}
    \beta^{\sys{E}}_\epsilon\infdivc*{\channel{N}_{\sys{AE}\to\sys{B}}}{\channel{M}_{\sys{AE}\to\sys{B}}}{\rho_{\sys{AR}}} 
    & = \inf_{0\mle M \mle I_\sys{B}} \left\{ \beta(\set{B}^{\rho_{\sys{AR}}},M) \middle\vert \alpha(\set{A}^{\rho_{\sys{AR}}},M)\leq\epsilon \right\} \\
    \label{eq:beta DH relation jammer:2}
    & = \sup_{\bar \rho \in \set{A}^{\rho_{\sys{AR}}}, \bar \sigma \in \set{B}^{\rho_{\sys{AR}}}} \beta_\epsilon\infdiv*{\bar \rho}{\bar \sigma}\\
    & = \sup_{\sigma_\sys{E},\,\omega_\sys{E} \in \density(\hilbert_\sys{E})} \beta_{\epsilon}\infdiv*{\channel{N}_{\sys{AE}\to\sys{B}}(\rho_\sys{AR}\tensor\sigma_\sys{E})}{\channel{M}_{\sys{AE}\to\sys{B}}(\rho_\sys{AR}\tensor\omega_\sys{E})},
\end{align}
where~\eqref{eq:beta DH relation jammer:2} follows from~\cite[Lemma~31]{fang2024generalized}.
Combining with~\eqref{eq:beta DH relation jammer:1}, we have 
\begin{align}
    \beta^{\sys{E}}_\epsilon\infdiv*{\channel{N}_{\sys{AE}\to\sys{B}}}{\channel{M}_{\sys{AE}\to\sys{B}}}
    = \adjustlimits \inf_{\rho_\sys{AR} \in \density(\hilbert_\sys{A}\tensor\hilbert_\sys{R})} \sup_{\sigma_\sys{E},\,\omega_\sys{E} \in \density(\hilbert_\sys{E})} \beta_{\epsilon}\infdiv*{\channel{N}_{\sys{AE}\to\sys{B}}(\rho_\sys{AR}\tensor\sigma_\sys{E})}{\channel{M}_{\sys{AE}\to\sys{B}}(\rho_\sys{AR}\tensor\omega_\sys{E})}.
\end{align}
We have the asserted result by taking $-\log$ on both sides.
\end{proof}

Proposition~\ref{prop: beta DH relation jammer} connects the optimal type-II error to the hypothesis testing minimax channel divergence.
Therefore, the error exponent of question can be expressed as 
\begin{equation}
    \lim_{n\to\infty} -\frac{1}{n}\log{\beta^{\sys{E}}_\epsilon\infdiv*{\channel{N}_{\sys{AE}\to\sys{B}}^{\tensor n}}{\channel{M}_{\sys{AE}\to\sys{B}}^{\tensor n}}}
    = \lim_{n\to\infty} \frac{1}{n} \hDiv{\epsilon}^{\minimax}\infdiv*{\channel{N}_{\sys{AE}\to\sys{B}}^{\tensor n}}{\channel{M}_{\sys{AE}\to\sys{B}}^{\tensor n}} = \ ?
\end{equation}
A \emph{weak} asymptotic characterization is shown in the theorem below.
\begin{shaded}
\begin{theorem}(Stein's lemma.)\label{thm:stein's lemma jammer}
For any $\channel{N}\in\CPTP(\sys{AE}:\sys{B})$ and $\channel{M}\in\CP(\sys{AE}:\sys{B})$, 
\begin{align}
    \lim_{\epsilon\to0^+} \lim_{n\to\infty} -\frac{1}{n}\log{\beta^{\sys{E}}_\epsilon\infdiv*{\channel{N}_{\sys{AE}\to\sys{B}}^{\tensor n}}{\channel{M}_{\sys{AE}\to\sys{B}}^{\tensor n}}}
    = \uDiv^{\minimax,\infty}\infdiv*{\channel{N}_{\sys{AE}\to\sys{B}}}{\channel{M}_{\sys{AE}\to\sys{B}}}.
\end{align}
\end{theorem}
\end{shaded}
\begin{proof}
By Proposition~\ref{prop: beta DH relation jammer}, it is equivalent to prove that 
\begin{align}
    \lim_{\epsilon\to0^+}\lim_{n\to\infty} \frac{1}{n} \hDiv{\epsilon}^{\minimax}\infdiv*{\channel{N}_{\sys{AE}\to\sys{B}}^{\tensor n}}{\channel{M}_{\sys{AE}\to\sys{B}}^{\tensor n}}
    = \uDiv^{\minimax,\infty}\infdiv*{\channel{N}_{\sys{AE}\to\sys{B}}}{\channel{M}_{\sys{AE}\to\sys{B}}}.
\end{align}
By the first inequality in Fact~\ref{fact: DH petz sandwiched} (along with the data processing inequality of the Petz-R\'enyi divergence), we have for any $\alpha, \epsilon\in(0,1)$,
\begin{align}
    \mrDiv{\alpha}\infdiv*{\cdot}{\cdot}+ \frac{\alpha}{\alpha-1} \log \frac{1}{\epsilon}
    \leq \hDiv{\epsilon}\infdiv*{\cdot}{\cdot}
    \leq \frac{1}{1-\epsilon}(\uDiv\infdiv*{\cdot}{\cdot} + h(\epsilon)).
\end{align}
Hence, 
\begin{align}
    \mrDiv{\alpha}^{\minimax}\infdiv*{\channel{N}_{\sys{AE}\to\sys{B}}}{\channel{M}_{\sys{AE}\to\sys{B}}}+ \frac{\alpha}{\alpha-1} \log \frac{1}{\epsilon}
    & \leq \hDiv{\epsilon}^{\minimax}\infdiv*{\channel{N}_{\sys{AE}\to\sys{B}}}{\channel{M}_{\sys{AE}\to\sys{B}}}\\
    & \leq \frac{1}{1-\epsilon}\left(\uDiv^\minimax\infdiv*{\channel{N}_{\sys{AE}\to\sys{B}}}{\channel{M}_{\sys{AE}\to\sys{B}}} + h(\epsilon)\right).
\end{align}
Applying the above to $n$-fold channels and dividing by $n$, and taking $n\to \infty$, then $\epsilon \to 0$, 
\begin{align}
    \mrDiv{\alpha}^{\minimax,\infty}\infdiv*{\channel{N}_{\sys{AE}\to\sys{B}}}{\channel{M}_{\sys{AE}\to\sys{B}}} \leq
    \lim_{\epsilon\to0^+}\lim_{n\to\infty} \frac{1}{n} \hDiv{\epsilon}^{\minimax}\infdiv*{\channel{N}_{\sys{AE}\to\sys{B}}^{\tensor n}}{\channel{M}_{\sys{AE}\to\sys{B}}^{\tensor n}}
    \leq \uDiv^{\minimax,\infty}\infdiv*{\channel{N}_{\sys{AE}\to\sys{B}}}{\channel{M}_{\sys{AE}\to\sys{B}}}.
\end{align}
As this holds for any $\alpha \in (0,1)$, we can take the supremum over $\alpha$, and replace the left-most expression by
\begin{align}
    \sup_{\alpha\in(0,1)} \mrDiv{\alpha}^{\minimax,\infty}\infdiv*{\channel{N}_{\sys{AE}\to\sys{B}}}{\channel{M}_{\sys{AE}\to\sys{B}}}
    = \mDiv^{\minimax,\infty}\infdiv*{\channel{N}_{\sys{AE}\to\sys{B}}}{\channel{M}_{\sys{AE}\to\sys{B}}}
    = \uDiv^{\minimax,\infty}\infdiv*{\channel{N}_{\sys{AE}\to\sys{B}}}{\channel{M}_{\sys{AE}\to\sys{B}}}
\end{align}
where the equalities follows from Lemma~\ref{lem: DM alpha jam infty continuity} and~Lemma~\ref{lem: DM D jam infty}, respectively.
Therefore, we have 
\begin{align}
    \uDiv^{\minimax,\infty}\infdiv*{\channel{N}_{\sys{AE}\to\sys{B}}}{\channel{M}_{\sys{AE}\to\sys{B}}} \leq
    \lim_{\epsilon\to0^+}\lim_{n\to\infty} \frac{1}{n} \hDiv{\epsilon}^{\minimax}\infdiv*{\channel{N}_{\sys{AE}\to\sys{B}}^{\tensor n}}{\channel{M}_{\sys{AE}\to\sys{B}}^{\tensor n}}
    \leq \uDiv^{\minimax,\infty}\infdiv*{\channel{N}_{\sys{AE}\to\sys{B}}}{\channel{M}_{\sys{AE}\to\sys{B}}}.
\end{align}
which concludes the asserted result.
\end{proof}
\begin{remark}
    Note that when $\dim{\hilbert_\sys{E}} = 1$, the above theorem recovers the best-case channel discrimination result in~\cite[Theorem 3]{wang2019resource}. When $\dim{\hilbert_\sys{A}} = 1$, it recovers a weak version~\footnote{The previous result in~\cite[Theorem 1]{fang2025adversarial} has the strong converse property.} of the worst-case (adversarial) channel discrimination result in~\cite[Theorem 1]{fang2025adversarial} under non-adaptive strategies. As the strong converse property (i.e., the convergence without the dependence of $\epsilon$) of the best-case channel discrimination is still open~\cite{fang2025towards}, it is also unclear whether the strong converse property holds in the minimax setting considered here.
\end{remark} 

\section{Conclusion and Discussion} \label{sec:conclusion}
In this work, we studied quantum channel discrimination in the presence of jammers, formulating the problem as a competative minimax game between the tester and the adversary, each optimizing their respective strategies.
We established an asymptotic characterization of the optimal type-II error exponent under parallel strategies, showing that it is governed by the minimax channel divergence introduced herein. We also explored fundamental properties of the minimax channel divergence, which can be of independent interest and lay the foundation for future exploration.

\paragraph{Squential strategies.} An interesting direction for future research is the study of sequential strategies, where both the tester and the adversary may adapt their inputs based on the outputs of the channel and its complementary channel, respectively. Previous results have shown that adaptive strategies do not improve the Stein's exponent in either the best-case~\cite{fang2020chain} or worst-case~\cite{fang2025adversarial} channel discrimination scenarios. Whether adaptive strategies can provide an advantage in the competative minimax setting considered here remains an open question.


\paragraph{Weaker jammers.}
In the hypothesis testing setup considered here, the jammer’s state $\sigma_\sys{E}$ is chosen independently in~\eqref{eq:def:type-1:error} and~\eqref{eq:def:type-2:error}. 
Operationally, this corresponds to an extreme case that the adversary can tailor its jamming strategy to the underlying hypothesis -- for instance, when the adversary is the one who manufactures the device. 

It is also natural to consider a weaker model in which the jammer is \emph{oblivious} to the hypothesis. 
In this case, the hypothesis testing problem becomes
\begin{align*}
\tilde{\beta}^{\sys{E}}_\epsilon&\infdiv*{\channel{N}_{\sys{AE}\to\sys{B}}}{\channel{M}_{\sys{AE}\to\sys{B}}}
\\
& = 
\inf_{\substack{0\mle M \mle I_\sys{B}\\ \rho_\sys{AR} \in \density(\hilbert_\sys{A}\tensor\hilbert_\sys{R})}}
\sup_{\sigma_\sys{E}\in\density(\hilbert_\sys{E})}
\Big\{ 
    \tr\!\left[ \channel{M}_{\sys{AE}\to\sys{B}}(\rho_\sys{AR}\tensor\sigma_\sys{E}) M \right]
\;\Big|\;
    \tr\!\left[ \channel{N}_{\sys{AE}\to\sys{B}}(\rho_\sys{AR}\tensor\sigma_\sys{E})(I-M) \right] \leq \epsilon 
\Big\}.
\end{align*}

A natural guess of the corresponding divergence is the \emph{minimum channel divergence} (replacing the worst-case channel divergence $\iinfDiv$), defined as
\[
    \infDiv\infdiv*{\channel{N}_{\sys{A}\to\sys{B}}}{\channel{M}_{\sys{A}\to\sys{B}}} 
    \defeq 
    \inf_{\rho_\sys{A} \in \density(\hilbert_\sys{A})} 
    \Div\infdiv*{\channel{N}_{\sys{A}\to\sys{B}}(\rho_\sys{A})}{\channel{M}_{\sys{A}\to\sys{B}}(\rho_\sys{A})}.
\]
Most arguments in Section~\ref{sec:divergence} and Section~\ref{sec:parallel:discrimination} can be adapted in a relatively straightforward manner, with the following three exceptions:
\begin{enumerate}
    \item A new super-additivity property of the minimum channel divergence, analogous to~\cite[Lemma~21]{fang2024generalized}, is required for~\eqref{eq:superadditive:mrDiv:4}.
          Unfortunately, this does not hold in general.
    \item The continuity of the minimum channel divergence, analogous to~\cite[Lemma~22]{fang2024generalized}, is required for~\eqref{eq:DM alpha jam continuity:3}.
          We believe this generalization should be relatively straightforward.
    \item An analogue of~\cite[Lemma~31]{fang2024generalized} is needed for~\eqref{eq:beta DH relation jammer:2}, but this remains unclear at present.
\end{enumerate}
We leave the exploration of weaker jammers and these related questions for future work.


\bigskip
\paragraph{Acknowledgements.}
K.F. is supported by the National Natural Science Foundation of China (Grant No.~92470113 and~12404569), the Shenzhen Science and Technology Program (Grant No.~JCYJ20240813113519025), the Shenzhen Fundamental Research Program (Grant
No.~JCYJ2 0241202124023031), the 1+1+1 CUHK-CUHK(SZ)-GDST Joint Collaboration Fund
(Grant No. GRD P2025-022), and the University Development Fund (Grant No. UDF01003565).
M.C. is supported by the European Research Council (ERC Grant Agreement No.~948139) and the Excellence Cluster Matter and Light for Quantum Computing (ML4Q). 

\bibliographystyle{ieeetr}
\bibliography{Bib, extra_reference}

\appendix
\section{Generalized version for Lemma~\ref{lem:mimmaxDiv:perm}}
\label{app:lem:bipartite:covarian}
\begin{lemma}\label{lem:bipartite:covariant}
Let $\Div\infdiv{\cdot}{\cdot}$ be a quantum divergence with direct sum property.
Let $\mathcal{N} \in \CPTP(\sys{AE}:\sys{B})$ and $\mathcal{M} \in \CP(\sys{AE}:\sys{B})$.
If they are jointly covariant with respect to $\{\mathcal{U}_A(g)\otimes\tilde{\mathcal{U}}_E(g), \mathcal{V}_B(g)\}_{g\in\mathcal{G}}$ for some group $\mathcal{G}$ where $\{\mathcal{U}_A(g)\}_{g\in\mathcal{G}}$, $\{\tilde{\mathcal{U}}_E(g)\}_{g\in\mathcal{G}}$, and $\{\mathcal{V}_B(g)\}_{g\in\mathcal{G}}$ are unitary representations of $\mathcal{G}$ on system $A$, $E$, and $B$, respectively, then one can restrict the domain of $\sigma$ and $\omega$ from $\density(E)$ to density operators invariant under the action of $\mathcal{G}$ in the following optimization
\begin{equation}
\adjustlimits\inf_{\sigma\,\omega \in \density(E)} \sup_{\rho\in \density(A)} \Div(\mathcal{N}_{AE\to B}(\phi_{AR}^{\rho} \otimes \sigma)\|\mathcal{M}_{AE\to B}(\phi_{AR}^{\rho} \otimes \omega)).
\end{equation}
\end{lemma}
\begin{proof}
The idea is that the set of density operators $\density(A)$ is invariant under unitary transformations, and that $\Div$ is also unitary invariant, hence for all $g\in\mathcal{G}$, 
\begin{align}
&\adjustlimits\inf_{\sigma\,\omega \in \density(E)} \sup_{\rho\in \density(A)} \Div(\mathcal{N}_{AE\to B}(\phi_{AR}^{\rho} \otimes \sigma)\|\mathcal{M}_{AE\to B}(\phi_{AR}^{\rho} \otimes \omega))\\
& = \adjustlimits\inf_{\sigma\,\omega \in \density(E)} \sup_{\rho\in \density(A)} \Div(\mathcal{N}_{AE\to B}\circ \mathcal{U}_A^{-1}(g)(\phi_{AR}^{\rho} \otimes \sigma)\|\mathcal{M}_{AE\to B}\circ \mathcal{U}_A^{-1}(g)(\phi_{AR}^{\rho} \otimes \omega))\\
& = \adjustlimits\inf_{\sigma\,\omega \in \density(E)} \sup_{\rho\in \density(A)} \Div(\mathcal{V}_B(g)\circ\mathcal{N}_{AE\to B}\circ \mathcal{U}_A^{-1}(g)(\phi_{AR}^{\rho} \otimes \sigma)\|\mathcal{V}_B(g)\circ\mathcal{M}_{AE\to B}\circ \mathcal{U}_A^{-1}(g)(\phi_{AR}^{\rho} \otimes \omega))\\
& = \adjustlimits\inf_{\sigma\,\omega \in \density(E)} \sup_{\rho\in \density(A)} \Div(\mathcal{N}_{AE\to B}\circ \tilde{\mathcal{U}}_E(g)(\phi_{AR}^{\rho} \otimes \sigma)\|\mathcal{M}_{AE\to B}\circ \tilde{\mathcal{U}}_E(g)(\phi_{AR}^{\rho} \otimes \omega))
\end{align}
Note that $\Div(\cdot\|\cdot)$ is convex (which is implied by data processing inequality and direct sum property), and that $\mathcal{N}$ and $\mathcal{M}$ are linear.
Hence, for any optimizer $(\sigma,\omega)$, one can define a group invariant optimizing pair as
\begin{equation}
(\sigma^\star,\omega^\star) 
=  \left(\sum_{g\in\mathcal{G}} \frac{1}{|\mathcal{G}|} \tilde{\mathcal{U}}_E(g)(\sigma), \sum_{g\in\mathcal{G}} \frac{1}{|\mathcal{G}|} \tilde{\mathcal{U}}_E(g)(\omega)\right).
\end{equation}
Hence, one can restrict the domain for $\sigma$ and $\omega$ to group invariant ones.
\end{proof}
\end{document}